\documentclass{article}

\def\noheaderplainsetup{

\topmargin=0pt \headheight=0pt \headsep=0pt  \oddsidemargin=0pt \evensidemargin=0pt  \textheight=8.9truein \textwidth=6.5truein}   

\noheaderplainsetup

\usepackage{amsfonts}

\begin{document}


\newcommand{\clthree}{\mbox{\bf CL12}}
\newcommand{\cltw}{\mbox{\bf CL12}}
\newcommand{\clfour}{\mbox{\bf CL4}}
\newcommand{\arfour}{\mbox{\bf CLA4}} 
\newcommand{\arfive}{\mbox{\bf CLA5}} 
\newcommand{\arsix}{\mbox{\bf CLA6}}
\newcommand{\arseven}{\mbox{\bf CLA7}}
\newcommand{\areight}{\mbox{\bf CLA8}}
\newcommand{\arnine}{\mbox{\bf CLA9}}
\newcommand{\arten}{\mbox{\bf CLA10}}
\newcommand{\pa}{\mbox{\bf PA}}


\newcommand{\intimpl}{\mbox{\hspace{2pt}$\circ$\hspace{-0.14cm} \raisebox{-0.043cm}{\Large --}\hspace{2pt}}}
\newcommand{\zero}{\mbox{\small {\bf 0}}}
\newcommand{\one}{\mbox{\small {\bf 1}}}
\newcommand{\successor}{\mbox{\hspace{1pt}\boldmath $'$}}

\newcommand{\elz}[1]{\mbox{$\parallel\hspace{-3pt} #1 \hspace{-3pt}\parallel$}} 
\newcommand{\elzi}[1]{\mbox{\scriptsize $\parallel\hspace{-3pt} #1 \hspace{-3pt}\parallel$}}
\newcommand{\emptyrun}{\langle\rangle} 
\newcommand{\oo}{\bot}            
\newcommand{\pp}{\top}            
\newcommand{\xx}{\wp}               
\newcommand{\legal}[2]{\mbox{\bf Lr}^{#1}_{#2}} 
\newcommand{\win}[2]{\mbox{\bf Wn}^{#1}_{#2}} 
\newcommand{\seq}[1]{\langle #1 \rangle} 
\newcommand{\code}[1]{\ulcorner #1 \urcorner}          


\newcommand{\pst}{\mbox{\raisebox{-0.01cm}{\scriptsize $\wedge$}\hspace{-4pt}\raisebox{0.16cm}{\tiny $\mid$}\hspace{2pt}}}
\newcommand{\pcost}{\mbox{\raisebox{0.12cm}{\scriptsize $\vee$}\hspace{-4pt}\raisebox{0.02cm}{\tiny $\mid$}\hspace{2pt}}}

\newcommand{\gneg}{\mbox{\small $\neg$}}                  
\newcommand{\mli}{\hspace{2pt}\mbox{\small $\rightarrow$}\hspace{2pt}}                      
\newcommand{\cla}{\mbox{$\forall$}}      
\newcommand{\cle}{\mbox{$\exists$}}        
\newcommand{\mld}{\hspace{2pt}\mbox{\small $\vee$}\hspace{2pt}}     
\newcommand{\mlc}{\hspace{2pt}\mbox{\small $\wedge$}\hspace{2pt}}   
\newcommand{\mlci}{\hspace{2pt}\mbox{\footnotesize $\wedge$}\hspace{2pt}}   
\newcommand{\ade}{\mbox{\large $\sqcup$}}      
\newcommand{\ada}{\mbox{\large $\sqcap$}}      
\newcommand{\add}{\hspace{2pt}\mbox{\small $\sqcup$}\hspace{2pt}}                     
\newcommand{\adc}{\hspace{2pt}\mbox{\small $\sqcap$}\hspace{2pt}} 
\newcommand{\adci}{\hspace{2pt}\mbox{\footnotesize $\sqcap$}\hspace{2pt}}              
\newcommand{\clai}{\forall}     
\newcommand{\clei}{\exists}        
\newcommand{\tlg}{\bot}               
\newcommand{\twg}{\top}               
\newcommand{\fintimpl}{\mbox{\hspace{2pt}$\bullet$\hspace{-0.14cm} \raisebox{-0.058cm}{\Large --}\hspace{-6pt}\raisebox{0.008cm}{\scriptsize $\wr$}\hspace{-1pt}\raisebox{0.008cm}{\scriptsize $\wr$}\hspace{4pt}}}
\newcommand{\col}[1]{\mbox{$#1$:}}


\newtheorem{theoremm}{Theorem}[section]
\newtheorem{factt}[theoremm]{Fact}
\newtheorem{definitionn}[theoremm]{Definition}
\newtheorem{lemmaa}[theoremm]{Lemma}
\newtheorem{propositionn}[theoremm]{Proposition}
\newtheorem{conventionn}[theoremm]{Convention}
\newtheorem{examplee}[theoremm]{Example}
\newtheorem{exercisee}[theoremm]{Exercise}
\newtheorem{thesiss}[theoremm]{Thesis}
\newtheorem{remarkk}[theoremm]{Remark}
\newenvironment{definition}{\begin{definitionn} \em}{ \end{definitionn}}
\newenvironment{theorem}{\begin{theoremm}}{\end{theoremm}}
\newenvironment{lemma}{\begin{lemmaa}}{\end{lemmaa}}
\newenvironment{fact}{\begin{factt}}{\end{factt}}
\newenvironment{proposition}{\begin{propositionn} }{\end{propositionn}}
\newenvironment{convention}{\begin{conventionn} \em}{\end{conventionn}}
\newenvironment{example}{\begin{examplee} \em}{\end{examplee}}
\newenvironment{thesis}{\begin{thesiss} \em}{\end{thesiss}}
\newenvironment{remark}{\begin{remarkk} \em}{\end{remarkk}}
\newenvironment{exercise}{\begin{exercisee} \em}{\end{exercisee}}
\newenvironment{proof}{ {\bf Proof.} }{\  \rule{2.5mm}{2.5mm} \vspace{.2in} }
\newenvironment{idea}{ {\bf Proof idea.} }{\  \rule{1.5mm}{1.5mm} \vspace{.15in} }
\newenvironment{subproof}{ {\em Proof.} }{\  \rule{2mm}{2mm} \vspace{.1in} }

\title{Introduction to clarithmetic III}
\author{Giorgi Japaridze}

\date{}
\maketitle

\begin{abstract} The present paper constructs three new systems of clarithmetic (arithmetic based on {\em computability logic}): {\bf CLA8}, {\bf CLA9} and {\bf CLA10}. System {\bf CLA8} is shown to be sound and
 extensionally complete with respect to {\bf PA}-provably recursive time computability. This is in the sense that an arithmetical problem $A$ has a $\tau$-time solution for some {\bf PA}-provably recursive function $\tau$ iff $A$ is represented by some theorem of {\bf CLA8}. System {\bf CLA9} is shown to be sound and intensionally complete with respect to constructively {\bf PA}-provable computability. This is in the sense that a sentence $X$ is 
a theorem of  {\bf CLA9} 
iff, for some particular machine $\cal M$, {\bf PA} proves that $\cal M$ computes (the problem represented by) $X$. And system {\bf CLA10} is shown to be 
sound and intensionally complete with respect to not-necessarily-constructively {\bf PA}-provable computability. This means that a sentence $X$ is 
a theorem of {\bf CLA10} iff {\bf PA} proves that $X$ is computable, even if {\bf PA} does not ``know'' of any particular machine $\cal M$ that computes $X$.  
\end{abstract}


\



\section{Introduction}\label{intr}
 
Being a continuation of \cite{cla4} and \cite{cla5}, this article  relies on the terminology, notation, conventions and technical results of its predecessors, with which the reader is assumed to be well familiar. While the present paper is not self-contained, the entire ``Introduction to clarithmetic'' series {\em is} so, and can be read without prior familiarity with {\em computability logic} (CoL), which serves as a logical basis for all theories elaborated in the series.   

The previously constructed systems $\arfour$, $\arfive$, $\arsix$ and $\arseven$ form a sequence of incrementally powerful theories, sound and extensionally complete with respect to polynomial time computability, polynomial space computability, elementary recursive time (=space) computability, and primitive recursive time (=space) computability, respectively. Continuing that pattern, the present paper  introduces three new, incrementally strong (and stronger than their predecessors) theories $\areight$, $\arnine$ and $\arten$.

A natural extreme beyond primitive recursive time is {\em $\pa$-provably recursive time} (which can be easily seen to be equivalent to $\pa$-provably recursive space). That means considering $\pa$-provably recursive functions instead of primitive recursive functions as time complexity bounds for computational problems. Our present theory $\areight$ turns out to be sound and  complete with respect to $\pa$-provably recursive time computability in the same sense as $\arseven$ is sound and complete with respect to primitive recursive time computability. 
Remember that, on top of the standard Peano axioms, $\arseven$  had the single extra-Peano axiom $\ada x\ade y(y= x+ 1)$, and its only  nonlogical rule of inference, termed ``$\arseven$-Induction'',  was 
\[\frac{F(0)\hspace{30pt} F(x)\mli F(x+ 1)}{F(x)},\]
with no restrictions on $F(x)$. $\areight$ augments $\arseven$ through the following single additional rule
\begin{equation}\label{jul18a}
\frac{F(x)\add\gneg F(x)\hspace{30pt} \cle x F(x)}{\ade x F(x)},\end{equation}
where $F(x)$ is elementary. A justification for this rule is  that, if we know how to decide the predicate $F(x)$ (the left premise), and we also know that the predicate is true of at least one number (the right premise), then we can  apply the decision procedure to $F(0)$, $F(1)$, $F(2)$, \ldots until a number $n$ is hit such that the procedure finds $F(n)$ true, after which the conclusion $\ade xF(x)$ can be solved by choosing $n$ for $x$ in it. 
 
The story does not end with provably recursive time computability though.  Not all computable problems have recursive (let alone provably so) time complexity bounds. In other words, not all computable problems are {\bf recursive time computable}. An example is 
\begin{equation}\label{jul19a}
\ada x\bigl(\cle y \hspace{1pt}p(x,y)\mli \ade y\hspace{1pt}p(x,y)\bigr),\end{equation}
 where $p(x,y)$ is a decidable binary predicate such that the unary predicate $\cle y\hspace{1pt}p(x,y)$ is undecidable (for instance, $p(x,y)$ means ``Turing machine $x$ halts within $y$ steps'').  Problem (\ref{jul19a}) is solved by the following effective strategy: Wait till Environment chooses a value $m$ for $x$. After that, for $n=0,1,2,\ldots$, figure out whether $p(m,n)$ is true. If and when you find an $n$ such that $p(m,n)$ is true, choose $n$ for $y$ in the consequent and retire. On the other hand, if there was a recursive bound $\tau$ for the time complexity of a solution $\cal M$ of (\ref{jul19a}), then the following would be a decision procedure for (the undecidable) $\cle y\hspace{1pt}p(x,y)$: Given an input $m$ (in the role of $y$), run $\cal M$ for $\tau(|m|+ 1)$ steps in the scenario where Environment chooses $m$ for $x$ at the very beginning of the play of (\ref{jul19a}), and does not make any further moves. If, during this time, $\cal M$ chooses a number $n$ for $y$ in the consequent of (\ref{jul19a}) such that $p(m,n)$ is true, accept; otherwise reject.\footnote{An alternative solution: Figure out whether there is a number $n$ with $|n|\leq \tau(|m|+ 1)$ such that $p(m,n)$ is true. If yes, accept; otherwise reject.}  
 
A next natural step on the road of constructing incrementally strong clarithmetical theories for incrementally weak concepts of computability is to go beyond $\pa$-provably recursive time computability and consider the weaker concept of {\em constructively $\pa$-provable computability} of (the problem represented by) a sentence $X$. The latter means existence of a machine $\cal M$ such that $\pa$ proves that $\cal M$ computes $X$, even if the running time of $\cal M$ is not bounded by any recursive function. System $\arnine$ turns out to be sound and complete with respect to this sort of computability. That is, a sentence $X$ is provable in $\arnine$ if and only if it is constructively $\pa$-provably computable.  Deductively, $\arnine$ only differs from $\areight$ in that, instead of (\ref{jul18a}), it has the following, stronger, rule: 
\begin{equation}\label{jul18b}
\frac{F(x)\add\gneg F(x)}{\cle x F(x)\mli \ade x F(x)},\end{equation}
where $F(x)$ is elementary. Note that  (\ref{jul18b}) merely ``modifies'' (\ref{jul18a}) by changing the status of $\cle x F(x)$ from being a premise 
of the rule to being an antecedent of the conclusion. A justification for (\ref{jul18b}) is that, if we know how to decide the predicate $F(x)$, then we can apply the decision procedure to $F(0)$, $F(1)$, $F(2)$, \ldots until (if and when) a number $n$ is hit such that the procedure finds $F(n)$ true, after which the conclusion can be solved by choosing $n$ for $x$ in its consequent.  Note that, unlike the earlier-outlined strategy for (\ref{jul18a}), the present strategy may  look for $n$ forever, and thus never make a move. This, however, only happens when $\cle xF(x)$ is false, in which case the conclusion is automatically won. 

A further weaker concept of (simply) {\em $\pa$-provable computability} is obtained from that  of constructively $\pa$-provable computability by dropping the ``constructiveness'' condition. Namely, $\pa$-provable computability of a sentence $X$ means that $\pa$ proves that a machine $\cal M$ solving $X$ {\em exists}, yet without necessarily being able to prove ``$\cal M$ solves $X$'' for any {\em particular} machine $\cal M$. An example of a sentence that is $\pa$-provably computable yet not constructively so is $S\add\gneg S$, where $S$ is an elementary sentence with $\pa\not\vdash S$ and $\pa\not\vdash \gneg S$, such as G\"{o}del's sentence. Let $\cal L$ be a machine that chooses the left disjunct of $S\add\gneg S$ and retires. Similarly, let $\cal R$ be a machine that chooses the right disjunct and retires. One of these two machines is a solution of $S\add\gneg S$, and, of course, $\pa$ ``knows'' this. Yet, $\pa$ does not ``know'' which one of them is a solution (otherwise either $S$ or $\gneg S$ would be provable); nor does it have a similar sort of ``knowledge'' for any other particular machine.  

A system sound and complete with respect to $\pa$-provable computability is $\arten$. It augments $\arnine$ through the following additional rule: 
\begin{equation}\label{jul18c}
\frac{\cle xF(x)}{\ade x F(x)},\end{equation}
where the premise is an elementary sentence. The admissibility of this rule, simply allowing us to change $\cle x$ to $\ade x$, is obvious in view of the restriction that $\cle xF(x)$ is a sentence (that is, $F(x)$ contains no free variables other than $x$). Indeed, if an $x$ satisfying $F(x)$ exists, then it can as well be ``computed'' (generated), even if we do not know what particular machine ``computes'' it. As we remember, systems $\arfour$-$\arseven$ are sound in a strong, constructive sense. Specifically, there is an effective (in fact, efficient) procedure for extracting solutions from proofs. The same strong form of soundness holds for $\areight$ and $\arnine$ as well. $\arten$ stands out as the only system whose soundness theorem is not (and cannot be) constructive. Namely, while $\arten$-provability of a sentence $X$ implies that $X$ has an algorithmic solution, generally there is no effective  way to extract a particular solution from a proof of $X$.     

\section{Technical preliminaries}
All  terminology and notation not redefined in this paper has the same meaning as in \cite{cla4,cla5}. And all of our old conventions from \cite{cla4,cla5} extend to the present context as well.
Namely, as in \cite{cla5}, 
 a ``{\bf sentence}'' always means a sentence (closed formula) of the language of $\arfour$. Similarly for ``{\em formula}'', unless otherwise specified or suggested by the context. Also, where $n$ is a natural number, the {\bf standard term} for $n$ means  $0$ followed by $n$ ``$\successor$''s (e.g., $0\successor\successor\successor$ is the standard term for $3$). We may not always be very careful about terminologically or notationally differentiating between a number and the standard term for it.  

As a binary predicate of the variable $\cal X$ over HPMs and the variable $X$ over sentences, ``{\em $\cal X$ wins $X$}'' is not arithmetical (is not expressible in the language of $\pa$), for otherwise so would be the truth predicate for elementary sentences: such a sentence is true iff it is won by an HPM that makes no moves. 
 Remember from Section 14.3 of \cite{cla4} that, on the other hand, for any {\em particular}    sentence $X$, the (now unary) predicate ``$\cal X$ wins  $X$'' {\em is} arithmetical. Throughout this paper, for each sentence $X$, we assume the presence of a fixed elementary formula $\mathbb{W}^X(x)$ naturally representing such a predicate, and when  we say something like ``$\pa$ proves that $\cal X$ wins $X$'', what we precisely mean is that 
$\pa\vdash \mathbb{W}^X(\code{\cal X})$, where $\code{\cal X}$ is the standard term for the code of $\cal X$. Or, if we say ``$\pa$ proves that $X$ is computable'', what we precisely  mean is that $\pa\vdash \cle x\mathbb{W}^X(x)$.

The soundness proofs found in this article  rely on the following lemma. In it, a {\bf clarithmetical sequent} means a sequent $E_1,\ldots,E_n\intimpl E_0$ where each $E_i$ ($0\leq i\leq n$) is a sentence (of the language of $\arfour$). 
  
\begin{lemma}\label{j15a}   
 There is an  efficient procedure that takes an arbitrary \ $\cltw$-proof of an arbitrary clarithmetical sequent  {\em $E_1,\ldots,E_n\intimpl F$} 
 and constructs an $n$-ary  \mbox{GHPM}  $\cal M$   such that $\pa$ proves that, for any 
$n$-ary   \mbox{GHPMs}  ${\cal N}_1,\ldots,{\cal N}_n$,  if  each ${\cal N}_i(\code{{\cal N}_1},\ldots,\code{{\cal N}_n})$ ($1\leq i\leq n$) is a solution of $E_{i}$, then ${\cal M}(\code{{\cal N}_1},\ldots,\code{{\cal N}_n})$ is a solution of $F$. 
\end{lemma}

\begin{proof} If the  phrase ``$\pa$ proves that'' is deleted in the present lemma, we get nothing but  
a weak/simplified version of Theorem 10.1 of  \cite{cla4}, which, in turn, is a reproduction of Theorem 10.5 of \cite{Japlbcs}.  An
analysis of the proof of the latter, combined with some basic experience in working with $\pa$, reveals that the latter can be formalized in $\pa$ in 
the form required by our present lemma.   
\end{proof}

\section{$\areight$, a theory of {\bf PA}-provably recursive time computability}\label{ss11}

The language of each of the theories $\areight$, $\arnine$ and $\arten$ introduced in this paper is the same as that of any other system of clarithmetic constructed in the present series of articles --- that is, it is an extension of the language of $\pa$ through the additional binary connectives $\adc,\add$ and quantifiers $\ada,\ade$.  

The axiomatization of $\areight$ is obtained from that of $\arseven$ by adding the single new rule of inference, which we call {\bf Finite Search} ({\bf FS}): 
\[\frac{\ada \bigl(F(x)\add\gneg F(x)\bigr)\hspace{30pt} \ada\cle x F(x)}{\ada\ade x F(x)},\]
where $F(x)$ is any  elementary formula.\footnote{The forthcoming soundness theorem for $\areight$  would just as well go through without the requirement that $F(x)$ is elementary. But why bother: completeness can be achieved even with the present, restricted, form of FS.}

To summarize, the nonlogical axioms  of $\areight$ are those of $\pa$ (Axioms 1-7 from Section 11 of \cite{cla4}) plus one single additional axiom $\ada x\ade y(y= x\successor)$ (Axiom 8). There are no logical axioms. The only logical rule of inference  is Logical Consequence (LC) as defined in Section   10 of \cite{cla4}, and the only nonlogical rules of inference are FS and $\arseven$-Induction 
\[\frac{\ada F(0)\hspace{30pt} \ada\bigl(F(x)\mli F(x\successor)\bigr)}{\ada F(x)}.\]

We fix 
\[\mathbb{T}(x,y,z,t)\]
as a standard formula of the language of $\pa$ saying that Turing machine (encoded by) $x$, on input $y$, at computation step $z$, halts with output $t$. By an {\bf explicit $\pa$-provably recursive function} we mean a natural number $\tau$ such that, where $\code{\tau}$ is the standard term for it,  $\pa\vdash \cla y\cle z\cle t \mathbb{T}(\code{\tau},y,z,t)$. Context permitting, we usually identify such a number $\tau$ or term $\code{\tau}$ with the  (unary) function  computed by the Turing machine encoded by $\tau$. When $\tau$ is an explicit $\pa$-provably recursive function and a given HPM $\cal M$ runs in time $\tau$, we say that $\tau$ is an {\bf explicit $\pa$-provably recursive bound} for the time complexity of $\cal M$;  whenever such a $\tau$ exists, we say that 
$\cal M$ is a {\bf $\pa$-provably recursive time machine}.

\begin{theorem}\label{ptt1}
An arithmetical problem has a $\pa$-provably recursive time solution iff it is provable in $\areight$. 

Furthermore, there is an efficient procedure that takes an arbitrary extended $\areight$-proof of an arbitrary sentence $X$ and constructs a   
 solution of $X$ (of $X^\dagger$, that is) together with an explicit $\pa$-provably recursive bound for its time complexity. 
\end{theorem}

\subsection{Proof of the soundness part of Theorem \ref{ptt1}}

Consider any sentence $X$ with a fixed extended $\areight$-proof. Our goal is to  construct a $\pa$-provably recursive time HPM $\cal M$ such that $\cal M$ wins $X$ and, furthermore, $\pa$ proves that $\cal M$ wins $X$. We shall not explicitly address the question on the efficiency of our construction, because it is achieved the same way as in all previous soundness proofs. That is, as in the similar proofs of \cite{cla5},  we will limit ourselves to verifying the pre-``furthermore'' part of the theorem. 

We proceed by induction on the length of the proof of $X$. In each case of our induction, we show how to construct the above-mentioned HPM $\cal M$, together with a recursive function $\tau$, and present an informal  proof of the fact that $\cal M$ solves $X$ in time $\tau$. 
A reader sufficiently familiar with $\pa$ will immediately see that such a proof can be reproduced in $\pa$.

The case of $X$ being an axiom is simple and is handled in the same way as in the earlier soundness proofs. 
So is the case of $X$ being derived by LC, only now it relies on Lemma \ref{j15a} instead of the earlier relied-upon Theorem 10.1 of \cite{cla4} (= Theorem 10.5 of \cite{Japlbcs}). The case of $X$ being obtained by $\arseven$-Induction is essentially handled in the way as  in the soundness proof for $\arseven$ found in \cite{cla5}, only with the words ``primitive recursive'' replaced by ``$\pa$-provably recursive''.

So, the only case worth considering is that of $X$ being derived by FS. Assume $X$ is (the $\ada$-closure) of $\ade xF(x)$, and thus its premises are (the $\ada$-closures of) $F(x)\add \gneg F(x)$ and $\cle xF(x)$. 
By the induction hypothesis, there is an HPM ${\cal N}$ that solves 
$F(x)\add \gneg F(x)$. Similarly for the other premise $\cle xF(x)$ but, since the latter is elementary, its solvability (that is, the solvability of 
$\ada \cle xF(x)$) simply means that  $\cla \cle xF(x)$ is true.

To describe our purported solution $\cal M$ of $\ade xF(x)$, assume $x,\vec{v}$ are exactly the free variables of $F(x)$, so that   $F(x)$ can be rewritten as $F(x,\vec{v})$. For simplicity, we rule out the trivial case of $F(x)$ having no free occurrences of $x$. 
At the beginning,   $\cal M$  waits for Environment to choose constants for the free variables $\vec{v}$ of $\ade x F(x,\vec{v})$.   Assume $\vec{c}$ are the constants chosen for $\vec{v}$. From now on, we shall write $F'(x)$ for $F(x,\vec{c})$. Further, where $i$ is a natural number, we shall write ${\cal N}_{i}$ for the machine that works just like ${\cal N}$ does in the scenario where the adversary, at the beginning of the play, has chosen the constant $i$ for the variable $x$ and the constants $\vec{c}$ for the variables $\vec{v}$. So, ${\cal N}_{i}$ wins the constant game $F'(i)\add\gneg F'(i)$.

Environment's initial moves bring the original $\ada \ade xF(x)$ down to $\ade xF'(x)$. The goal of $\cal M$ now is to win $\ade xF'(x)$. It achieves this goal by creating a record $i$, initializing it to $0$, and then acting as prescribed by the following procedure:\vspace{5pt}

{\bf Procedure} LOOP: Simulate ${\cal N}_i$ until it chooses one of the two $\add$-disjuncts of $F'(i)\add\gneg F'(i)$. If the right disjunct is chosen, increment $i$ by $1$ and repeat LOOP. Otherwise, if the left disjunct is chosen, specify $x$ as $i$ in the (real) play of $\ade xF'(x)$, and retire. \vspace{5pt}

Since $\cle x F'(x)$ is true, sooner or later the above procedure hits an $i$ such that the simulated ${\cal N}_i$ chooses the left disjunct of $F'(i)\add\gneg F'(i)$, meaning that $F'(i)$ is true. This guarantees that $\cal M$ wins. A bound  $\tau$ for the time complexity of $\cal M$ is computed  by a Turing machine that follows the work of $\cal M$ and counts the steps that it performs before making a move. Such a function $\tau$ is $\pa$-provably recursive because, as already noted, our entire argument can be reproduced in $\pa$.
 
\subsection{Proof of the completeness part of Theorem \ref{ptt1}}
Consider an arbitrary sentence $X$, an arbitrary HPM $\cal X$, and an arbitrary explicit $\pa$-provably recursive function $\chi$ such that $\cal X$ is a $\chi$ time solution of $X$. Let $\code{\chi}$ be the standard term for (the code of) $\chi$. We fix some enumeration of pairs of natural numbers and, where $a$ is a natural number, denote the first (resp. second) element of the $a$th pair by $(a)_1$ (resp. $(a)_2$). We treat $(x)_1$, $(x)_2$ as pseudoterms and assume that the above enumeration is ``standard enough'', so that the functions $(x)_1$ and $(x)_2$ are primitive recursive and $\pa$ proves  
\begin{equation}\label{jul19b}
\cla \Bigl(\cle z_1\cle z_2\mathbb{T}\bigl(x,y,z_1,z_2\bigr)\leftrightarrow \cle z\mathbb{T}\bigl(x,y,(z)_1,(z)_2\bigr)\Bigr)
\end{equation}
($E\leftrightarrow F$ abbreviates $(E\mli F)\mlc (F\mli E)$). The function $\chi$ will also be treated as a pseudoterm. Namely, if we write $z=\chi(x)$ within a formal expression, it is to be understood as an abbreviation of $\cle y\mathbb{T}\bigl(\code{\chi},x,y,z)$. 

That $\chi$ is an explicit $\pa$-provably recursive function, by definition, means that $\pa$ proves 
\begin{equation}\label{jul19c}
\cla y\cle z_1\cle z_2\mathbb{T}(\code{\chi},y,z_1,z_2).
\end{equation}   
In $\areight$, from (\ref{jul19b}) and (\ref{jul19c}), by LC we  get
\begin{equation}\label{j13a}
 \ada \cle z\mathbb{T}\bigl(\code{\chi},y,(z)_1,(z)_2\bigr).
\end{equation} 

It is obvious that $\pa$ constructively proves (in the sense of Section 11 of \cite{cla5}) the primitive recursive time computability of 
\begin{equation}\label{j13b}
\ada \Bigl(\mathbb{T}\bigl(\code{\chi},y,(z)_1,(z)_2\bigr)\add\gneg \mathbb{T}\bigl(\code{\chi},y,(z)_1,(z)_2\bigr)\Bigr).
\end{equation}
Therefore, by Theorem 11.2 of \cite{cla5}, $\arseven$ proves (\ref{j13b}), and hence so does $\areight$ because the latter is an extension of the former.

From  (\ref{j13b}) and (\ref{j13a}), by FS, we get $\ada \ade z \mathbb{T}\bigl(\code{\chi},y,(z)_1,(z)_2\bigr)$ which, together with (\ref{jul19b}), by LC, can be easily seen to imply 
\begin{equation}\label{a2a}
\areight\vdash \ade z \bigl(z= \chi(x)\bigr).
\end{equation}

The rest of our completeness proof for $\areight$ is literally the same as the completeness proof for $\arsix$ found in Section 7 of \cite{cla5}, with the only difference that now $\chi$ is a $\pa$-provably recursive (rather than elementary recursive) function; also, where Section 7 of \cite{cla5} relied on Fact 7.2, now we  rely on (\ref{a2a}) instead.

\subsection{The intensional strength of $\areight$} 
We say that $\pa$ {\bf constructively proves the $\pa$-provably recursive time computability of} a sentence $X$ iff, for some particular HPM $\cal X$ and some particular explicit $\pa$-provably recursive function $\chi$, $\pa$ proves that $\cal X$ is   a $\chi$-time   solution of $X$.  

The following theorem holds for virtually the same reasons as the similar Theorem 16.2 of \cite{cla4} for $\arfour$ or Theorem 11.2 of \cite{cla5} for $\arfive$, $\arsix$ and $\arseven$:

\begin{theorem}\label{jan30}
Let $X$ be any sentence  such that $\pa$ constructively proves  the $\pa$-provably recursive time computability of $X$. Then $\areight$ proves $X$. 
\end{theorem}

\begin{remark} From our soundness proof for $\areight$ it is immediately clear that the above theorem, in fact, holds in the stronger, ``if and only if'' form. The same can be seen to be the case for  Theorem 16.2 of \cite{cla4} and Theorem 11.2 of \cite{cla5}.
\end{remark}

\section{$\arnine$, a theory of constructively $\pa$-provable computability}\label{ss12}

Deductively, $\arnine$ only differs from $\areight$ in that, instead of the FS rule of the latter, $\arnine$ has the following rule, which we call  {\bf Infinite Search} ({\bf IS}):

\[\frac{\ada \bigl(F(x)\add \gneg F(x)\bigr)}{\ada \bigl(\cle xF(x)\mli \ade xF(x)\bigr)},\]
where $F(x)$ is any elementary formula.

Let $X$ be a sentence. We say that $\pa$ {\bf constructively proves the computability of $X$} iff, for some HPM $\cal X$, $\pa$ proves that $\cal X$ wins $X$.

Notice that, while our earlier defined concepts of polynomial, elementary recursive, primitive recursive and $\pa$-provably recursive time (or space) computabilities are {\em extensional} in their nature, the concept of constructively $\pa$-provable computability is {\em intensional}. To be precise, the former are properties of computational {\em problems} while the latter is a property of {\em sentences}. An extensional version of this concept could be defined by saying that a computational problem $A$ is constructively $\pa$-provably computable in the extensional sense iff there is a sentence $X$ with $X^\dagger=A$ ($X$ ``represents'' $A$) such that $X$ is constructively $\pa$-provably computable in the intensional sense. The forthcoming Theorem \ref{mainth} can be easily seen to continue to hold after replacing the intensional concept of constructively $\pa$-provable computability by its extensional counterpart.\footnote{On the other hand, Theorem \ref{ptt1} and similar theorems from \cite{cla4,cla5} --- namely, their completeness parts --- would fail with the intensional counterparts of the corresponding extensional concepts of computability.} Yet, doing so would significantly and unnecessarily weaken the theorem. This is the reason why, in the present context, we have opted for only considering the intensional concept.  
A similar comment applies to the concept of (not-necessarily-constructively) $\pa$-provable computability defined later in Section \ref{ss12}.

\begin{theorem}\label{mainth}
For any sentence $X$, $\arnine$ proves $X$ iff $\pa$ constructively proves the computability of $X$.

Furthermore, there is an efficient procedure that takes an arbitrary extended $\arnine$-proof of an arbitrary sentence $X$ and constructs an HPM $\cal X$ such that $\pa$ proves that $\cal X$ is a solution of $X$ (of $X^\dagger$, that is).
\end{theorem}

\subsection{Proof of the soundness part of Theorem \ref{mainth}}\label{sectsound}

As in the earlier soundness proofs, we will limit ourselves to verifying the pre-``furthermore'' part. Consider any sentence $X$ with a fixed extended $\arnine$-proof.  We proceed by induction on the length of the proof of $X$. In each case of the induction, we (show how to) construct an HPM $\cal M$ and present an informal  proof of the fact that $\cal M$ solves $X$. It will be immediately clear that  such a proof can be reproduced in $\pa$.

The case of $X$ being an axiom is simple and is handled as in the soundness proof for $\areight$. 
So is the case of $X$ being derived by LC.
So is the case of $X$ being derived by $\arseven$-Induction --- namely, it is essentially handled in the same way as in the soundness proof for $\arseven$ found in \cite{cla5}, but is, in fact, simpler, because we no longer need to care about complexity. 

Finally, assume $X$ is (the $\ada$-closure of) $\cle xF(x)\mli \ade xF(x)$, obtained by IS from (the $\ada$-closure of)  $F(x)\add \gneg F(x)$. By the induction hypothesis, we know how to construct a solution $\cal N$ of  $F(x)\add \gneg F(x)$. Here we assume that $x$  occurs free in $F(x)$ (otherwise the case is straightforward), and that $\vec{v}$ are all the additional free variables of $F(x)$. So, $F(x)$ can be rewritten as $F(x,\vec{v})$. We let $\cal M$ --- the purported solution of $\cle xF(x)\mli \ade xF(x)$ --- be a machine that, at the beginning of the play, waits till Environment chooses constants $\vec{c}$ for the variables $\vec{v}$. 
Let $F'(x)$ stand for $F(x,\vec{c})$ and, for each natural number $i$, let ${\cal N}_i$ be an HPM that plays $F(x)\add \gneg F(x)$ just as $\cal N$ does in the scenario where, at the very beginning of the play, the adversary chose the constant $i$ for the variable $x$ and the constants $\vec{c}$ for the variables $\vec{v}$. 

After the above event of Environment having chosen constants for all free variables, thus having brought the original game $\ada \bigl(\cle xF(x)\mli\ade xF(x)\bigr)$ down to $\cle xF'(x)\mli\ade xF'(x)$,  $\cal M$ creates a record $i$, initializes it to $0$, and then acts as prescribed by the following procedure:\vspace{5pt}

{\bf Procedure} LOOP: Simulate ${\cal N}_i$ until it chooses one of the two $\add$-disjuncts of $F'(i)\add\gneg F'(i)$. If the right disjunct is chosen, increment $i$ by $1$ and repeat LOOP. Otherwise, if the left disjunct is chosen, specify $x$ as $i$ in the consequent of $\cle xF'(x)\mli \ade xF'(x)$,  and retire. \vspace{5pt}

If $\cle x F'(x)$ is false, $\cle xF'(x)\mli \ade xF'(x)$ is (automatically) won by $\cal M$. And if $\cle x F'(x)$ is true, sooner or later the above procedure hits an $i$ such that the simulated ${\cal N}_i$ chooses the left disjunct of $F'(i)\add\gneg F'(i)$; then, again,  $\cal M$ wins, because it brings the consequent of $\cle xF'(x)\mli \ade xF'(x)$ down to the true $F'(i)$. 

\subsection{Proof of the completeness part of Theorem \ref{mainth}}\label{s19}

Consider an arbitrary sentence $X$ and an arbitrary HPM $\cal X$. Let $\mathbb{L}$ be an elementary sentence saying ``$\cal X$ does not win $X$'', so that $\gneg \mathbb{L}$ says ``$\cal X$ wins $X$''. Our intermediate --- and main --- goal is to show that $\arnine\vdash \gneg \mathbb{L}\mli X$ (Lemma \ref{july}), from which the desired completeness of $\arnine$ follows almost immediately. For the purposes of the subsequent section, it is important to note that, at this point, we are not making any assumptions about $X$ and $\cal X$.
In particular, we are {\em not} assuming that $\cal X$ wins $X$, let alone that $\pa$ proves so; such an assumption will be made only later, {\em after} Lemma \ref{july} is proven.

By a {\bf computation history} we shall mean a finite initial segment of some computation branch of $\cal X$. 
The way we encode configurations of $\cal X$ is described in Appendix A.1 of \cite{cla4}. That encoding extends to computation histories as finite sequences of configurations in a standard way. 
For readability, we will often identify configuration histories with their codes and say something like ``$a$ is a computation history'' when what is precisely meant is ``$a$ is the code of a computation history''.  We may further identify such an $a$ with the standard term for it. 
 
Remember from Section 14.5 of \cite{cla4} that, where $\Phi$ is a legal position of $X$, the {\bf yield} of $\Phi$ means the game $\seq{\Phi}X$. So, the type of the ``yield'' function in  \cite{cla4} is $\{\mbox{\em positions}\}\times\{\mbox{\em games}\}\longrightarrow \{\mbox{\em games}\}$. Here, for safety, we need an ``intensional'' version of this concept/function, whose type is $\{\mbox{\em positions}\}\times\{\mbox{\em $\cltw$-formulas}\}\longrightarrow \{\mbox{\em $\cltw$-formulas}\}$. Namely, let $F$ be a closed $\cltw$-formula containing no predicate letters other than $=$, and no function letters other than $\successor,+,\times$ ($F$ is not necessarily a sentence of the language of $\arnine$ because it may contain constants other than $0$). And let $\Phi$ be a legal position of $F$ (of $F^\dagger$, that is). Then, in the context of $F$, the intensional version of the yield  of $\Phi$, denoted by $\seq{\Phi}!F$, is defined inductively as follows:
\begin{itemize}
  \item $\seq{}!F=F$ (remember that $\seq{}$ means the empty run).
  \item For any nonempty legal position $\seq{\lambda,\Psi}$ of $F$ (where $\lambda$ is a labmove and $\Psi$ is a sequence of labmoves): 
  \begin{itemize}  
  \item If $\lambda$ signifies a choice of a component $G_i$ in an occurrence of a subformula $G_0\add G_1$ or $G_0\adc G_1$ of $F$, and $F'$ is the result of replacing that occurrence by $G_i$ in $F$, then $\seq{\lambda,\Psi}! F=\seq{\Psi}!F'$.
  \item If $\lambda$ signifies a choice of a constant $c$ for a variable $x$ in an occurrence of a subformula $\ade xG(x)$ or $\ada xG(x)$ of $F$, and $F'$ is the result of replacing that occurrence by $G(c)$ in $F$, then $\seq{\lambda,\Psi}! F=\seq{\Psi}!F'$.  
\end{itemize}  
\end{itemize}

Let $E(\vec{s})$ be a   formula all of whose  free variables are among  $\vec{s}$ (but not necessarily vice versa), and let $z$ be a variable not among $\vec{s}$. We 
 will write   $\tilde{E}(z,\vec{s})$ 
to denote an elementary formula whose free variables are $z$ and those of $E(\vec{s})$,  and which is a natural arithmetization of the predicate that,  for any constants $a,\vec{c}$ in the roles of $z,\vec{s}$, holds (that is, $\tilde{E}(a,\vec{c})$ is true) iff $a$ is a computation history and, where $\Phi$ is the position spelled on the run tape of the last configuration of that history, $\Phi$ is a legal 
position of $X$ with  $\seq{\Phi}!X=E(\vec{c})$.

Let $y$ be a variable and $E$ be a formula not containing $y$. As in Section 14.6 of \cite{cla4}, we say that a formula $H$ is a {\bf $(\oo,y)$-development} of  $E$ iff $H$ is the result of replacing in $E$: 
\begin{itemize}
\item either a surface occurrence of a subformula $F_0\adc F_1$ by $F_i$ ($i=0$ or $i=1$), 
\item  or  a surface occurrence of a subformula $\ada xF(x)$ by $F(y)$.   
\end{itemize}

{\bf $(\pp ,y)$-development} is defined in the same way, only with $\add,\ade$ instead of $\adc,\ada$.

\begin{lemma}\label{m2a}
Assume $E(\vec{s})$ is a formula all of whose free variables are among $\vec{s}$, and $y,z$ are  variables not occurring in $E(\vec{s})$. Then: 

(a) $\arnine$ proves $\tilde{E}(z,\vec{s})\add\gneg\tilde{E}(z,\vec{s})$.     

(b) For every $(\oo,y)$-development $H_{i}(y,\vec{s})$ of $E(\vec{s})$,  $\arnine$ proves $\tilde{E}(z,\vec{s})  \mli \ade u \tilde{H}_{i}(u,y,\vec{s})$.

(c) Where $H_1(y,\vec{s}),\ldots,H_n(y,\vec{s})$ are all of the $(\pp,y)$-developments of $E(\vec{s})$, $\arnine$ proves
\begin{equation}\label{m2e}   \gneg\mathbb{L}\mlc \tilde{E}(z,\vec{s}) \mlc \gneg \elz{E(\vec{s})} \mli  \ade u \ade y \tilde{H}_{1}(u,y,\vec{s})\add\ldots\add\ade u \ade y \tilde{H}_{n}(u,y,\vec{s})  . \end{equation}
\end{lemma}

\begin{proof} {\em Clause (a)}: Observe that the predicate $\tilde{E}(z,\vec{s})$ is  primitive recursive and, furthermore, $\pa$ obviously constructively proves the primitive recursive time computability of $\tilde{E}(z,\vec{s})\add\gneg\tilde{E}(z,\vec{s})$. Hence, by Theorem 11.2 of \cite{cla5}, $\arseven$ proves 
$\tilde{E}(z,\vec{s})\add\gneg\tilde{E}(z,\vec{s})$. Therefore, of course, so does $\arnine$.\vspace{5pt}

{\em Clause (b)}: Pick any one of the $(\oo,y)$-developments $H_{i}(y,\vec{s})$ of $E(\vec{s})$.  By clause (a) of the present lemma with $H_{i}(y,\vec{s})$ in the role of $E(\vec{s})$ and $u$ in the role of $z$, $\arnine$ proves 
 $ \tilde{H}_{i}(u,y,\vec{s}) \add \gneg \tilde{H}_{i}(u,y,\vec{s})$, 
whence, by Infinite Search,  $\arnine$ also proves 
\begin{equation}\label{nov15a}
\cle u\tilde{H}_{i}(u,y,\vec{s})  \mli \ade u\tilde{H}_{i}(u,y,\vec{s}).
\end{equation}

Consider any  ($\cla$)  values  of $y$ and $\vec{s}$, and assume $\tilde{E}(z,\vec{s})$ is true. 
 $H_{i}(y,\vec{s})$ is the game to which $E (\vec{s})$ is brought down by a certain labmove $\oo\alpha$. Let 
 $u$  be the result of appending such a labmove $\oo\alpha$ to the run-tape content of the last configuration of $z$. Then, obviously, $\tilde{H}_{i}(u,y,\vec{s})$  is true.   To summarize, we have just found that, if $\tilde{E}(z,\vec{s})$ is true, then so is $\cle u\tilde{H}_{i}(u,y,\vec{s})$.

Of course, the argument of the preceding paragraph is  formalizable in $\pa$,  which implies the ($\pa$-  and hence)  $\arnine$-provability of  
$  \cla \bigl(\tilde{E}(z,\vec{s})\mli\cle u\tilde{H}_{i}(u,y,\vec{s})\bigr)$. But, as it is not hard to see, for any formula $F$, $\ada F$ is a logical consequence of $\cla F$. Hence, the sentence $\ada \bigl(\tilde{E}(z,\vec{s})\mli\cle u\tilde{H}_{i}(u,y,\vec{s})\bigr)$ is also provable in $\arnine$, which is the same as to say that $\arnine$ proves  
\begin{equation}\label{nov15b}
\tilde{E}(z,\vec{s})\mli\cle u\tilde{H}_{i}(u,y,\vec{s}).
\end{equation}

Now, the target $\tilde{E}(z,\vec{s})  \mli \ade u \tilde{H}_{i}(u,y,\vec{s})$ can be seen to be a logical consequence of   (\ref{nov15a}) and (\ref{nov15b}), which makes it provable in $\arnine$.\vspace{5pt}

{\em Clause (c)}: Unlike our handling of clauses (a) and (b) that involved some explicit metareasoning, for both brevity and diversity, in the present case we the higher-level  approach of informally (but formalizably) reasoning directly in $\arnine$. 

So, to justify (\ref{m2e}), argue in $\arnine$. Assume $\gneg \mathbb{L}$, $\tilde{E}(z,\vec{s})$ and $\gneg \elz{E(\vec{s})}$. Consider the scenario where Environment does not move in any configuration starting from the last configuration of $z$. If $\cal X$ does not move either, it can be seen to lose, because $E(\vec{s})$ is the final position reached in the play and its elementarization $\elz{E(\vec{s})}$, by our assumption, is false. But  our assumption $\gneg\mathbb{L}$ means 
that $\cal X$ cannot lose. So,  sooner or later,  $\cal X$ will move. Let us call the computation history that extends $z$ to the point when the above event of $\cal X$ moving happens {\bf magical}. Thus,  a magical computation history exists in the sense of $\cle$. Also,  for any particular computation history $u$, we can obviously tell --- in the sense of $\add$ --- whether $u$ is magical or not.\footnote{After all, the predicate of ``being magical'' is primitive recursive, so $\arseven$ is sufficient to decide it.}   Hence, by Infinite Search, a magical computation history  exists not only in the sense of $\cle$, but also in the sense of $\ade$. That is, it can be actually found/computed. So, let $m$ be such a magical computation history. From $m$, we can further find, in the sense of $\ade$, the particular move $\alpha$ that $\cal X$ made in the above-described scenario, i.e., in the scenario represented by $m$.   This $\alpha$ must be legal, or else $\cal X$ would (but cannot) lose. If so, $\alpha$  brings $E(\vec{s})$ down to $H_{i}(n,\vec{s})$ for a certain $n$ and certain  $(\pp,y)$-development $H_{i}(y,\vec{s})$ of $E(\vec{s})$.  Again,  such  $i$ and $n$ not only exist, but can be actually found. Now, we can win (\ref{m2e}) by  choosing the  $\add$-disjunct $\ade u\ade y \tilde{H}_{i}(u,y,\vec{s})$  in the consequent, and then, in it,  specifying $u$ as $m$ and $y$ as $n$.\vspace{-7pt} 
\end{proof}

\begin{lemma}\label{m2c}
Assume $z,\vec{s}$ are pairwise distinct variables, and $E(\vec{s})$ is a formula all of whose free variables are among $\vec{s}$. 
Then \[\arnine\vdash \gneg \mathbb{L}\mlc\tilde{E}(z,\vec{s})  \mli E(\vec{s}).\]  
\end{lemma}

\begin{proof} We prove this lemma by induction on the complexity of $E(\vec{s})$. Pick a fresh variable $y$. By the induction hypothesis, for any $(\oo,y)$- or $(\pp,y)$-development $H_i(y,\vec{s})$ of $E(\vec{s})$ (if there are any), $\arnine$ proves 
\begin{equation}\label{m2d}
\gneg \mathbb{L}\mlc \tilde{H}_{i}  (u,y,\vec{s})  \mli H_i(y,\vec{s}).
\end{equation}

Argue in $\arnine$ to justify $\gneg \mathbb{L}\mlc\tilde{E}  (z,\vec{s})  \mli E(\vec{s})$. Consider any values (constants) $b$ and $\vec{a}$ chosen by Environment for 
$z$ and $\vec{s}$, respectively.\footnote{Here, unlike the earlier followed practice, for safety, we are reluctant to use the names $z$ and $\vec{s}$ for those constants.}  Throughout the rest of this argument, assume both $\gneg \mathbb{L}$ and $\tilde{E}  (b,\vec{a})$ are true (otherwise we win). We need to see 
how to win  $E(\vec{a})$. 
 
To solve $E(\vec{a})$, we bring the resource (\ref{m2e}) down to 
\[\gneg \mathbb{L}\mlc \tilde{E}(b,\vec{a}) \mlc \gneg \elz{E(\vec{a})} \mli  \ade u\ade y \tilde{H}_{1}(u,y,\vec{a})\add\ldots\add\ade u\ade y \tilde{H}_{n}(u,y,\vec{a}).\] 
Since  the $\gneg \mathbb{L}$ and $\tilde{E}(b,\vec{a})$ components of the above are true, in fact, the following resource is at our disposal:  
\begin{equation}\label{may18a}
\gneg \elz{E(\vec{a})} \mli  \ade u\ade y \tilde{H}_{1}(u,y,\vec{a})\add\ldots\add\ade u\ade y \tilde{H}_{n}(u,y,\vec{a}).\end{equation} 

We wait till one of the following two events takes place:

{\em Event 1}: Environment makes a move $\alpha$ in $E(\vec{a})$. We may assume that this move is legal. Then, for one of the $(\oo,y)$-developments $H_i(y,\vec{s})$ of $E(\vec{s})$ and some constant $c$, the labmove $\oo\alpha$ brings   $E(\vec{a})$  down to $H_i(c,\vec{a})$. So, now it remains to see how to win $H_i(c,\vec{a})$. In view of clause (b) of Lemma \ref{m2a} and the truth of $\tilde{E} (b,\vec{a})$, the resource $ \ade u \tilde{H}_{i}(u,c,\vec{a})$ is at our disposal. Using it, we find a $d$ with $ \tilde{H}_{i}(d,c,\vec{a})$. Now we bring (\ref{m2d}) down to 
$\gneg \mathbb{L}\mlc \tilde{H}_{i}  (d,c,\vec{a})  \mli H_i(c,\vec{a})$. Since the antecedent of this resource is true, it provides a sought way to win  
$H_i(c,\vec{a})$. 

{\em Event 2}: The provider of (\ref{may18a}) brings it down to  
\begin{equation}\label{jul29} 
 \gneg \elz{E(\vec{a})} \mli  \tilde{H}_{i}(d,c,\vec{a}).\end{equation} 
for one of $i\in\{1,\ldots,n\}$ and some constants $d,c$. Using clause (a) of Lemma \ref{m2a}, we check whether $\tilde{H}_{i}(d,c,\vec{a})$ is true. If not, $\elz{E(\vec{a})}$ is guaranteed to be true (otherwise (\ref{jul29}) would be lost by its provider), and we continue waiting for Event 1 to occur: if such an event never occurs, the truth of $\elz{E(\vec{a})}$ obviously means that we win. Suppose now $\tilde{H}_{i}(d,c,\vec{a})$ is true. Then, just as in the case of Event 1, (\ref{m2d}) provides a way to win $H_i(c,\vec{a})$. We make the move that brings $E(\vec{a})$ down to $H_i(c,\vec{a})$, and follow the just-mentioned way. 

If neither event happens, then $\elz{E(\vec{a})}$ is true (otherwise (\ref{may18a}) would be lost by its provider) and, again, we win. 
\end{proof}

\begin{lemma}\label{july}
$\arnine\vdash \gneg \mathbb{L}\mli X$.
\end{lemma}

\begin{proof}
 Let $a$ be the code of the empty computation history, and $\hat{a}$ be the standard term for $a$.  Of course, $\pa$ and hence $\arnine$ proves $\tilde{X}(\hat{a})$. In view of Fact 12.6 of \cite{cla4}, $\arnine$ proves $\ade z(z=\hat{a})$. 
 By Lemma \ref{m2c},   $\arnine$ also proves $\ada z\bigl(\gneg\mathbb{L}\mlc \tilde{X}  (z)  \mli X\bigr)$.  These three can be seen to imply  $\gneg\mathbb{L}\mli X$ by LC. 
\end{proof} 
     
Now we are ready to claim the target result of this section. Suppose $\pa$ constructively proves the computability of $X$. We may assume that it proves ``$\cal X$ solves $X$'' for the earlier-fixed (yet arbitrary) HPM $\cal X$. In other words, $\pa\vdash\gneg\mathbb{L}$. Hence $\arnine\vdash\gneg\mathbb{L}$.  Then, in view of Lemma \ref{july}, $\arnine\vdash X$, as desired.

\section{$\arten$, a theory of {\bf PA}-provable computability}\label{ss12}

$\arten$ only differs from $\arnine$ in that it has the following additional rule of inference, called {\bf Constructivization}:

\[\frac{\cle xF(x)}{\ade xF(x)},\]
where $F(x)$ is any elementary\footnote{Just as in the case of $\areight$, the requirement that $F(x)$ is elementary can be dropped here without affecting the soundness of the system.} formula containing no free variables other than $x$.

Let $X$ be a sentence. We say that $\pa$ {\bf proves the computability of $X$} iff $\pa$ proves that there is ($\cle$) an HPM that 
wins $X$.

\begin{theorem}\label{mainthten}
For any sentence $X$, $\arten$ proves $X$ iff $\pa$ proves the computability of $X$.
\end{theorem}

\subsection{Proof of the soundness part of Theorem \ref{mainthten}}

As always, this part of the theorem is verified by induction on the length of the proof of $X$. The cases where $X$ is an axiom, or is derived by LC or IS, are handled in the same way as in the soundness proof for $\arnine$. Here we shall only look at the case of $X$ being derived by Constructivization. So, assume $X$ is $\ade xF(x)$, where $F(x)$ is an elementary formula not containing any free variables other than $x$, and  $\ade xF(x)$ is derived from $\cle xF(x)$ by Constructivization. By the induction hypothesis, $\pa$ proves that $\cle xF(x)$ is ``computable'' which, as $F(x)$ is elementary, simply means that $\cle xF(x)$ is true. So, $\pa\vdash \cle xF(x)$. 

Argue in $\pa$. Since $\cle xF(x)$ is true, there is a number $a$ such that $F(a)$ is true. Then the target $\ade xF(x)$ is solved by an HPM $\cal M$ that makes $a$ as its only move in the play and retires. We ($\pa$, that is) cannot name such an $\cal M$ because we do not know what number $a$ exactly is; yet we know that $\cal M$ {\em exists}. In other words, we know that $X$ is computable. 

\subsection{Proof of the completeness part of Theorem \ref{mainthten}}

Consider an arbitrary sentence $X$. Let $\mathbb{L}(x)$ be a natural formalization of the predicate ``$x$ is (the code of) an HPM which does not win $X$''. Remember that, in  Section \ref{s19}, $\cal X$ was a fixed yet arbitrary HPM. And note that the sentence $\mathbb{L}$ of Section \ref{s19} was nothing but what we can now write as $\mathbb{L}(\code{\cal X})$. So, Lemma \ref{july} can now be re-stated as 
\begin{equation}\label{julya}
\arnine\vdash \gneg \mathbb{L}(\code{{\cal X}})\mli X.
\end{equation} 
Further observe that, while $\cal X$ was fixed in Section  \ref{s19}, the proof of Lemma \ref{july} given there goes through with $\cal X$ as a ($\ada$-quantified) variable; more precisely,   (\ref{julya}), in fact, holds in the following, stronger form:
\begin{equation}\label{julyb}
\arnine\vdash \ada x\bigl(\gneg \mathbb{L}(x)\mli X\bigr).
\end{equation}
By LC, (\ref{julyb}) implies 
\begin{equation}\label{julyc}
\arnine\vdash \ade x\gneg \mathbb{L}(x)\mli X.
\end{equation}
Assume now $\pa$ proves that $X$ is computable. In other words, $\pa\vdash \cle x\gneg\mathbb{L}(x)$. Then, by Constructivization, $\arten\vdash \ade x\gneg\mathbb{L}(x)$. But then, 
in view of (\ref{julyc}), $\arten\vdash X$, as desired.


\vspace{10pt}  




\begin{thebibliography}{99}



\bibitem{Japtowards} G. Japaridze. {\em Towards applied theories based on computability logic}. {\bf Journal of Symbolic Logic} 75 (2010), pp. 565-601.  


\bibitem{cla4} G. Japaridze. {\em Introduction to clarithmetic I}.  {\bf Information and Computation} 209 (2011), pp. 1312-1354.

\bibitem{Japlbcs} G. Japaridze. {\em A logical basis for constructive systems}. {\bf Journal of Logic and Computation} 22 (2012), pp. 605-642.   


\bibitem{cla5} G. Japaridze. {\em Introduction to clarithmetic II}.  Manuscript at http://arxiv.org/abs/1004.3236
\end{thebibliography}
\end{document}